%% file: supplementary-material.tex
\documentclass[preprint,nocopyrightspace]{sigplanconf}

\usepackage{times}
\usepackage{listings}          
\usepackage{enumitem}      

\usepackage{amsmath}
\usepackage{amssymb}
\usepackage{amsthm}

\theoremstyle{plain}
\newtheorem{theorem}{Theorem}[section]
\newtheorem{lemma}[theorem]{Lemma}
\usepackage{hyperref}
\usepackage{graphicx}
\usepackage{algorithm}
\usepackage{algorithmic}
\usepackage{balance}
\usepackage{multirow}
\usepackage{tabularx,booktabs}
\usepackage{etex}
\usepackage{tikz}
\usetikzlibrary{positioning}
\usetikzlibrary{calc}
\usepackage{pifont}
\setcitestyle{numbers}
\begin{document}

\title{An Approach for Finding Permutations Quickly: Fusion and Dimension
matching}
\authorinfo{Aravind Acharya}
{Department of Computer Science\\ and Automation \\ Indian Institute of Science
\\ Bangalore 560012, India}
{aravind@iisc.ac.in}
\authorinfo{Uday Bondhugula}
{Department of Computer Science\\ and Automation \\ Indian Institute of Science
\\ Bangalore 560012, India}
{udayb@iisc.ac.in}
\authorinfo{Albert Cohen}
{INRIA and DI\\Ecole Normale Superieure \\ 45 Rue d'Ulm
\\ Paris 75230, France}
{Albert.Cohen@inria.fr}

\maketitle
\input{tex/pldi-supplementary-material-abstract.tex}
\input{tex/supplementary-background}

\input{tex/proofs.tex}

\setcounter{algorithm}{3}
\input{tex/align-scale.tex}

\balance
\bibliographystyle{plain}
\bibliography{bibfile}
\end{document}

%% file: tex/pldi-supplementary-material-abstract.tex
\begin{abstract}
Polyhedral compilers can perform complex loop optimizations that improve
parallelism and cache behaviour of loops in the input program. These
transformations result in significant performance gains on modern processors
which have large compute power and deep memory hierarchies. The paper,
\emph{Polyhedral Auto-transformation with No Integer Linear Programming},
identifies issues that adversely affect scalability of polyhedral
transformation frameworks; in particular the Pluto algorithm. The construction
and solving of a complex Integer Linear Programming (ILP) problem increases the
time taken by a polyhedral compiler significantly. The paper presents two
orthogonal ideas, which together overcome the scalability issues in the affine
scheduling problem. It first relaxes the ILP to a Linear Programming (LP)
problem, thereby solving a cheaper algorithm. To overcome the sub-optimalities
that arise due to this relaxation, the affine scheduling problem is decomposed
into following three components: (1) Fusion and dimension matching, (2) Loop
scaling and shifting, and (3) Loop skewing. This new auto-transformation
framework, pluto-lp-dfp, significantly improves the time taken by the Pluto
algorithm without sacrificing performance of the generated code. This report
first provides proofs for the theoretical claims made in the paper surrounding
relaxed LP formulation of the Pluto algorithm. The second part of the report
describes an approach to find good loop fusion (or distribution) and loop
permutations that enable tileability. This short report serves as the
supplementary material for the paper.
\end{abstract}

%% file: tex/supplementary-background.tex
\section{Background}
\label{sec:background}
\newcommand{\card}[1]{\vert {#1} \vert}
In this section, we introduce terminology used in the report. We also provide 
background on the current ILP formulation used in Pluto to find good
transformations. This is included for the purpose of completeness.

\subsection{Affine Transformations}

A polyhedral compiler framework has a statement-centric view of the program.
Each statement in an iteration space is modeled with integer sets called
index sets or the domain of the statement. Let $\mathbf{S}$ be the set of 
all statements. Let $I_S$ denote the index set of a particular
statement $S$. For example, if $S$ has a two-dimensional index set, then:
\begin{equation}
    I_S = \{[i,j] \mid 0\leq i,j\leq N-1\}
\end{equation}
defines an index set where $i$ and $j$ are the original loop iterator variables
of the statement $S$ and $N$ is a program parameter. These index sets 
represent the set of statement instances that are executed by the program.  
An instance of $S$ is given by the iteration vector of $S$ referred to as 
$\vec{i}_S$.  Let $m_S$ denote the dimensionality of $I_S$, i.e.,  
$\vec{i_S}$ has $m_S$ components corresponding to loops surrounding the
statement $S$ from outermost to innermost.

Data dependences are precisely represented in a polyhedral 
auto-transformation framework using dependence polyhedra, which are a 
conjunction of constraints. These constraints can also be viewed as a 
relation between source and target iterations. These relations include 
affine combinations of loop iterator variables of the source and target 
iterations, program parameters and existentially quantified variables. If 
$D_e$ is the dependence polyhedron associated with an edge $e$ of the data 
dependence graph, then an iteration $\vec{t}$ of a statement $S_j$ is 
dependent on an iteration $\vec{s}$ of a
statement $S_i$ if and only if $\langle\vec{s},\vec{t}\rangle \in D_e$.


Formally, an affine transformation in the polyhedral model is a 
multi-dimensional affine function of
the loop iterators and program parameters. A one-dimensional affine
transformation $\phi_S$ for the statement $S$ (corresponding to a particular 
{\it level} or loop depth roughly speaking) can be expressed as:
\begin{align*}
    \phi_S(\vec{i_S}) = (c_1,c_2,\dots,c_{m_S}).(\vec{i_S}) + (d_1,\dots
    d_p).(\vec{p})+ c_0,\\
    c_0,c_1,\dots c_{m_S}, d_1,\dots d_p \in \mathbb{Z}.
\end{align*}
Each statement has its own set of $c_i$'s and $d_i$'s, and these are called
transformation coefficients corresponding to loop iterator variables and
program parameters (denoted by $\vec{p}$) respectively. The transformation
$\phi_S$ at a level $i$ for a statement $S$, can also be viewed as a hyperplane
, denoted by $\vec{h_i^S}$.  For simplicity, we drop the statement identifier
$S$ in places where the meaning is clear from context. 
The set of
consecutive hyperplanes that can be permuted form a \emph{permutable band}.
These hyperplanes (from outermost to innermost) form the rows of the
transformation matrix. The term {\it schedule} and {\it transformation} are
also used interchangeably, since a transformation specifies a new schedule. 

\subsection{ILP Formulation in Pluto}
Pluto is a polyhedral optimizer that finds affine transformations to maximize 
locality and parallelism. Given the index sets of the statements in the 
program and the dependences in the form of dependence polyhedra, the Pluto 
algorithm iteratively finds linearly independent hyperplanes based on an 
objective that minimizes dependence distances. This objective is modeled 
using an ILP, which we describe in the rest of this section.

The Pluto algorithm iteratively finds hyperplanes from outermost to 
innermost looking for tileable bands, i.e., the hyperplanes that satisfy the
tiling validity constraint below for every
dependence $\langle\vec{s},\vec{t}\rangle \in D_e$:
\begin{equation}
    \label{eqn:tile-validity}
    \phi_{S_j}(\vec{t})-\phi_{S_i}(\vec{s})\geq 0,
\end{equation}
where $S_i$ and $S_j$ represents the source and target statements of $D_e$ 
respectively.

The objective used by the Pluto algorithm is then to minimize the
dependence distances using a bounding function:
\begin{equation}
    \label{eqn:dep-bound}
    \phi_{S_j}(\vec{t})-\phi_{S_i}(\vec{s})\leq \vec{u}.\vec{p} + w.
\end{equation}
The intuition behind this upper bound on dependence distances is as follows: 
the dependence distances are bounded by loop iterator variables, which are 
further bounded by program parameters. Therefore, one can choose large 
enough values for $\vec{u}$ to obtain an upper bound. In order to minimize 
dependence distances, the Pluto algorithm minimizes coefficients of this 
upper bound by finding the \emph{lexicographic minimum (lexmin)} of ($\vec{u}$, $w$) 
as the objective:
\begin{equation}
    \label{eqn:lexmin}
    \text{lexmin}\left(\vec{u},w,\dots,c_i^S,d_i^S,\dots\right),
\end{equation}
where $c_i^S$ and $d_i^S$ are the transformation coefficients of $S$.

Note that well-known ILP solvers like GLPK~\cite{glpk}, Gurobi~\cite{gurobi} 
and CPLEX do not provide a {\it lexmin} function. However, lexmin can be 
implemented in practice as a weighted sum objective, in which the coefficients of
$\vec{u}$ (in the objective function) will be significantly higer than the coefficient of $w$ and so on.
%
\subsubsection{Avoiding the zero solution}
The tiling validity constraints and the dependence bounding constraints from 
(\ref{eqn:tile-validity})~and~(\ref{eqn:dep-bound}) have a trivial zero
vector solution. The Pluto algorithm restricts all transformation 
coefficients (of $\phi_S's)$  to non-negative integers. This restriction 
allows us to avoid the trivial solution for the coefficients of $\phi_S$ 
with the constraint:
\begin{equation}
    \label{eqn:trivial-soln}
    \sum\limits_{i=0}^{m_S} c_i\geq 1.
\end{equation}

\subsubsection{Linear Independence}
Affine transformations have to be one-to-one mappings in order for them to 
specify a complete schedule.  The Pluto algorithm thus enforces linear 
independence of hyperplanes, statement-wise. This is modeled by finding a 
basis for the null space of hyperplanes already found.  The next hyperplane 
to be found must have a component in this null space. The exact modelling of
this constraint is described in~\cite{uday16toplas}. It will be a constraint 
of the form:
\begin{equation}
    \label{eqn:lin-ind}
    \sum\limits_{i=0}^{m_s} a_i \times c_i \geq 1,
\end{equation}
where $a_i \in \mathbb{Z}$. These $a_i$'s are from the subspace that is orthogonal
to the subspace of currently found hyperplanes. 
For the rest of this paper, we refer to the above formulation as {\it 
pluto-ilp}.

We denote the set of statements in the program with $\mathbf{S}$, and 
$\card{\mathbf{S}}$ is its cardinality. Similarly, $\mathbf{C}$ denotes the 
set of connected components in the data dependence graph (DDG). We use
$\psi$ to denote affine constraints on transformation coefficients, loop
bounds, dependence distances, and program parameters.

%% file: tex/proofs.tex
\section{Proofs}
\label{append:proofs}

The constraints shown in Equation~\ref{eqn:lin-ind-real} model the full space
of non negative rational solutions.  Eventhough these constraints can not be
implemented in the solver, we use these constraints to prove certain
interesting results that exist when linear indpendence and constraints are
modelled precisely.
\begin{equation}
    \sum\limits_{i=0}^{m_S} c_i > 0, \qquad
    \label{eqn:lin-ind-real}
    \sum\limits_{i=0}^{m_S} a_i \times c_i > 0.
\end{equation}

In Lemma~\ref{thm:rational-solution} and Theorem~\ref{thm:scaling-appendix}, we
prove properties of solutions of the relaxed formulation that hold
irrespective of the way linear independence and trivial solution avoiding
constraints are modeled (either using
Equation~\ref{eqn:trivial-soln}, ~\ref{eqn:lin-ind} or \ref{eqn:lin-ind-real}).
\begin{lemma}
    \label{thm:rational-solution}
    The optimal solution to \textit{pluto-lp}, when it exists, is {\it
    rational}.
\end{lemma}
The above lemma follows from the fact that all coefficients in the LP
formulation of Pluto are integers; thus the solutions of {\it pluto-lp} are
rational.  For the rest of this paper, we refer to the optimal rational
solution of pluto-lp as the solution of pluto-lp.

\begin{theorem}
    \label{thm:scaling-appendix}
    If $\vec{z}$ is a solution to the relaxed Pluto formulation ({\it
    pluto-lp}), then for any constant $k \geq 1$, $k\times\vec{z}$ is also a
    valid solution to {\it pluto-lp}.
\end{theorem}
\begin{proof}

    From Lemma~\ref{thm:rational-solution}, we know that, if a solution
    exists, then the optimal value of the objective corresponds to a rational
    solution of pluto-lp. Now, we need to prove that, scaling the solutions of
    pluto-lp will not violate the constraints. Consider the tiling validity
    constraints in (\ref{eqn:tile-validity}). $\phi_{S_i}$ and $\phi_{S_j}$ are
    one dimensional affine transformations. Therefore,

    \begin{align*}
        \phi_{S_j}(\vec{t})- \phi_{S_i}(\vec{s})\geq 0 \implies k\times
        \phi_{S_j}(\vec{t})- k\times \phi_{S_i}(\vec{s})\geq 0
    \end{align*}
where $k\geq 1$. Therefore any hyperplane found by the pluto-lp will not
violate the tiling validity constraints after scaling the solutions. The
dependence bounding constraints in~(\ref{eqn:dep-bound}) are bounded above by
$\vec{u}$ and $w$, which are variables in the pluto-lp formulation. The values
of $\vec{u}$ and $w$  are can also be scaled up without violating the
constraints.  That is,
    \begin{center}
        $k\times\phi_{S_j}(\vec{t})- k\times \phi_{S_i}(\vec{s})\geq k\times\vec{u}+k\times w$.
    \end{center}
    The trivial solution avoiding constraints and the linear independence
    avoiding constraints given in~(\ref{eqn:lin-ind-real}) can not be violated
    by scaling the solutions.  That is, if $c_i$'s are the solutions to
    pluto-lp and $k\geq 1$, then from (\ref{eqn:lin-ind-real}) it follows that
    for each statement $S$,
    \begin{equation*}
        \sum\limits_{i=0}^{m_S} k\times c_i > 0, \qquad 
        \sum\limits_{i=0}^{m_S} a_i \times k\times c_i > 0.
    \end{equation*}
    Therefore scaling the solutions of pluto-lp with a factor $k\geq 1$, will
    not violate the constraints.
\end{proof}

Theorems~\ref{thm:optimality}~and~\ref{thm:ratio-preservation} refer to the
constraints that hold only in cases where the linear independence and trivial
solution avoiding constraints model the full space of rational solutions as
given in Equation~\ref{eqn:lin-ind-real}.
\begin{theorem}
    \label{thm:optimality}
    The optimal solution to the relaxed Pluto algorithm ({\it pluto-lp}) can be
    scaled to an integral solution to {\it pluto-lp} such that the objective of
    the scaled (integral) solution will be equal to the objective of optimal
    solution of {\it pluto-ilp}.
\end{theorem}
\begin{proof}
    By Theorem~\ref{thm:scaling-appendix}, we know that after scaling the real
    solutions, the resulting solution does not violate any constraints.  Let
    $z_i$ and $z_r$ be the value of the optimal objective values for solutions
    to {\it pluto-ilp} and {\it pluto-lp} respectively.  Let $c_s$ be the
    smallest scaling factor that scales solutions of {\it pluto-lp} to
    integers. Let $z_i^\prime=c_s\times z_r$ be the value obtained from by
    scaling the optimal real solution of {\it pluto-lp} to an integral one.
    Note that $c_s \geq 1$; otherwise, the real solution would not be optimal.
    Now we prove that $z_i^\prime\leq z_i$.  Consider $z_r^\prime$ given by
    \begin{align*}
        &z_r^\prime =z_i/c_s \nonumber \\
        &\implies z_r\leq z_r^\prime \text{ ($\because z_r$ is the optimal solution to {\it pluto-lp})} \nonumber\\
        &\implies c_s\times z_r\leq c_s\times z_r^\prime \text{ ($\because c_s\geq 1$ and $z_r,z_r^\prime \geq 0$)} \nonumber\\
        &\implies z_i^\prime \leq z_i.
    \end{align*}
    This proves that the optimal (minimum) objective of {\it pluto-lp} after
    scaling to integer coefficients will be less than or equal to that of of
    {\it pluto-ilp}.  However, the objective of the relaxed formulation after
    scaling cannot be strictly less than that of {\it pluto-ilp} (otherwise,
    {\it pluto-ilp}'s solution, $z_i$ would not be an optimal one).  Therefore
    the optimal objective of {\it pluto-lp} after scaling up, is equal to the
    optimal objective of {\it pluto-ilp}.
\end{proof}
                                                                                                                                                    
\begin{theorem}
    \label{thm:ratio-preservation}
    Let $\vec{h_i}=(c_1,\dots, c_n)$ be the optimal solution for {\it
    pluto-ilp}.  Then, the optimal solution to {\it pluto-lp}, $\vec{h_r}$, is
    such that $\vec{h_r}=\vec{h_i}/c_s$ where $c_s\geq 1$.
\end{theorem}

\begin{proof}
    Let $z_i$ and $z_r$ be the optimal values of the objective found by {\it
    pluto-ilp} and {\it pluto-lp} respectively. Let $c_s$ be the smallest
    scaling factor that scales every component of $\vec{h_r}$ to an integer.
    Note that $c_s \geq 1$ (otherwise, $\vec{h_r}$ would not have been optimal
    solution).  Let $z_i$ be the solution obtained by the hyperplane
    $\vec{h_i}$.  Let $z_r^\prime=z_i/c_s$.  Note that $z_r^\prime$ can be
    obtained by dividing all the components of $\vec{h_i}$ by $c_s$.  Now we
    have the following cases:
    \begin{itemize}
        \item \textbf{Case 1:} If $z_r=z_r^\prime$ then we have nothing to prove.
        \item \textbf{Case 2:} Consider the case $z_r<z_r^\prime$.  Let
            $\vec{h_i^\prime}=\vec{h_r}\times c_s$, and let $z_i^\prime$ be the
            objective value with ${h_i^\prime}$.  $z_i^\prime=z_r\times c_s$
            (due to the nature of (\ref{eqn:dep-bound})).  Since the optimal
            objective value for {\it pluto-ilp} was found to be $z_i$,
            $z_i^\prime \geq z_i$.  Now if we scale down each component of
            $\vec{h_i}$ by $c_s$, we get a solution that has an objective value
            lower than $z_r$.  This is a contradiction.
    \end{itemize}
    Therefore, $z_r = z_r^\prime$ in all cases, and $z_i' = z_i$.  Since both
    $\vec{h_i}$ and $\vec{h_i'}$ have the same optimal objective value and
    given that the lexmin provides a unique optimal solution, $\vec{h_i'} =
    \vec{h_i}$, and $\vec{h_r} = \vec{h_i}/c_s$.
\end{proof}

The properties of the solutions of {\it pluto-lp} that hold even when
linear independece constraints and trivial solution avoiding constraints model
the space of rational solutions imprecisely are stated in
Theorems~\ref{thm:parallel}~and~\ref{thm:tileability}.
\begin{theorem}
    \label{thm:parallel}
    The relaxed formulation, {\it pluto-lp} (in each permutable band), finds a
    outer parallel hyperplane if and only if {\it pluto-ilp} finds a outer parallel
    hyperplane.
\end{theorem}
\begin{proof}
    There exists a parallel hyperplane if and only if $\vec{u}=\vec{0}$ and
    $w=0$ in the ILP formulation of Pluto. Note that $\vec{u}+w$ gives
    an upper bound on the dependence distance and therefore $\vec{u}+w$
    is the smallest value of $\vec{u} + w$. The objective of the relaxed LP is
    to minimize the values of $\vec{u}$ and $w$, and there exists an integer
    solution which is also present in the real space. Therefore the
    solution found by {\it pluto-lp} will fall into one of the two following
    cases.
    \begin{enumerate}
        \item The solution found by {\it pluto-lp} is same as the solution found
            by
            {\it pluto-ilp}. In this case, there is nothing to prove.  \item
        \label{proof:converse}{\it pluto-lp} finds a fractional solution
            with $\vec{u} = \vec{0}$ and
            $w=0$. In this case, by Theorem~\ref{thm:scaling-appendix} one can scale
            the
            real (fractional) solution to an integral one without violating
            any constraints.
            This scaling up will neither change the value of $\vec{u}$ nor $w$
            because they were found to be equal to zero.
        \end{enumerate}
        The ``only if'' part of the proof follows from Case~\ref{proof:converse}
        in the above argument.
\end{proof}

\begin{theorem}
\label{thm:tileability}
Given a loop nest of dimensionality $m$, if {\it pluto-ilp} finds $d \leq m$
permutable hyperplanes, then {\it pluto-lp}
also finds $d$ permutable hyperplanes.
\end{theorem}
\begin{proof}
Let us assume that {\it pluto-lp} finds $k$ hyperplanes and let $k \neq d$. We
prove
Theorem~\ref{thm:tileability} by contradiction. Let us assume that
$k>d$. The $k$ linearly independent hyperplanes found by {\it pluto-lp} can be
scaled to integers. These scaled solutions will continue to be linearly
independent as scaling transformations will not affect linear independence.
Therefore these correspond to $k$ linearly independent in the integer space.
This means that there existed $k$ linearly independent solutions in the integer
space. Since the validity constraints remain the same at each level, there
exits only $d$ linearly independent solutions as found by {\it pluto-ilp}.
This is a contradiction to the assumption that $k>d$.

Suppose $k<d$, then we know that there are $d$ linearly independent solutions
to the tiling validity constraints in the integer space. These are valid linearly
independent solutions in the rational space. Therefore, it is a contradiction
to our assumption $k<d$. Therefore, {\it pluto-lp} will find $d$ linearly
independent solutions to the tiling validity constraints.
\end{proof}

\subsection{Proofs corresponding to routines in {\it pluto-lp-dfp}} In this
section we state and prove theorems that establish the correctness of routines
in {\it pluto-lp-dfp} framework.  Algorithm~1 refers to the scaling MIP,
\textsc{Scale}, presented in Section 4; Algorithm~2 refers to the scaling and
shifting routine, \textsc{ScaleAndShiftPermutations} presented in Section~5 and
Algorithm~3 refers to the skewing routine, \textsc{IntroduceSkew}, presented in
Section~6 of the PLDI paper describing the framework of {\it pluto-lp-dfp}.

\begin{theorem}
    \label{thm:skew-disable}
    If loop skewing and shifting transformations are disabled, the relaxed Pluto algorithm will
    find the transformation coefficients that are scaled down versions of the
    transformation coefficients of the {\it pluto-ilp}. The values of $\vec{u}$ and
    $w$ in {\it pluto-lp} will be scaled down by the same scaling factor as
    the transformation coefficients.
\end{theorem}
\begin{proof}
    When loop skewing and shifting transformations are disabled, then only one of the
    transformation coefficients is non-zero. Without loss of
    generality, let us assume that $c_i$ is the non zero
    coefficient in each statement. The real space of non-zero solutions is modeled
    imprecisely. Therefore we can assume that the lower bound of $c_i$
    of every statement to be $1$. Let $\vec{u}.\vec{p}+w = z_i$ and
    $\vec{u}.\vec{p}+w = z_r$ for {\it pluto-ilp} and {\it pluto-lp} respectively.
    Let us normalize the values of $c_i$ to $1$. Let $c_s^\prime$ be
    the normalizing factor. The normalized coefficients,
    $c_i/{c_s^\prime}$ for each statement S, will be in the space of
    real solutions. This is because the lower bound of each $c_i$
    is $1$. The value of other $c_i$'s will also correspond to the
    lower bounds that are obtained Gaussian elimination and Fourier
    Motzkin elimination of the variables from validity and dependence
    bounding constraints. These correspond to the lowest possible
    values each of these $c_i$'s can take. Therefore the value of
    $z_i^\prime$ is equal to that of $z_r$. Hence the values of the
    variables in the ILP, including the objective, will be scaled down by the same scaling
    factor in the relaxed LP formulation.
\end{proof}

\begin{theorem}
Given a valid transformation $T$ for a program, the output transformation that
is obtained by Algorithm~3 does not violate any dependences.
\end{theorem}
\begin{proof}
The proof can be split into two cases: (1) If the algorithm did not introduce a
skew, then the we return the input transformation itself. Since the input
transformation did not violate any dependences, the returned transformation is
valid. (2) If the Algorithm~3 introduced a skew, then
for each level $i$, it only uses the transformation coefficients from the outer
levels. Note that the algorithm proceeds level by level. Let $i$ be the
dimension at which are introducing a skew. All other hyperplanes which are outer
to $i$ can be permuted to the outer level because, $i$ is the first dimension
that has a negative component for some dependence. The dimensions that are used
to skew will have coefficients from the outer levels and none of these levels
have a negative component. Since the newly introduced skew satisfies pluto-lp,
it does not violate any dependences. All dependences that were previously
satisfied by level $i$ will still continue to be satisfied at level $i$ after
skewing. Hence all
the dependences will be satisfied by the transformation obtained
from Algorithm~3.
\end{proof}

\begin{theorem}
    Given a valid transformation $T$ for a program,
Algorithm~3 does not introduce any skewing
transformations in cases where $T$ was tileable.
\end{theorem}
\begin{proof}
The algorithm tries to introduce a skew only when there is a negative component
for one of the dependences at a level $d$. This means that the level $d$ can
not be permuted to the outermost level. Hence the input transformation $T$
would result in a loop nest which can not be tiled. Therefore,
Algorithm~3 would not be introducing skewing, if the
original loop nest was not tileable.
\end{proof}

\begin{theorem}
Given a program $P$, the transformation hyperplanes obtained by {\it pluto-lp-dfp}
are linearly independent and do not violate any dependences.
\end{theorem}

\begin{proof}
The correctness claim of {\it pluto-lp-dfp} follows from correctness of
Algorithm~2~and~3.
Both algorithms find valid affine transforms and these transformations can be
composed together without violating any depenedences. In order the prove that
the found affine transformation hyperplanes are linearly independent, we prove
that each step in {\it pluto-lp-dfp} preserve linear independence of
hyperplanes. In the first step, linear independence of hyperplanes is first guaranteed by the
initial
permutation. Then scaling and shifting transformations introduced by
Algorithm~2 does not affect linear
independence. The skew introduced by Algorithm~3 will not affect linear
independence as well
because a skew introduced at each level will have a new component which
does not exist either in outer or inner levels. Therefore, the transformations hyperplanes
obtained from {\it pluto-lp-dfp} will be linearly independent. 
\end{proof}

%% file: tex/align-scale.tex
\renewcommand{\dim}[1]{\mathit{dim}(#1)}
\section{An approach for finding permutations quickly: Fusion and dimension
matching}
\label{sec:fusion-matching}

In this section we describe our approach to find a valid permutation. This is
the first step in {\it pluto-lp-dfp} after polyhedral dependences are obtained.
A permutation $\mathbb{P}$ is said to be valid if there are loop scaling and
loop shifting factors for (each dimension in $\mathbb{P}$), such that the
resulting transformation will not violate any dependences. The objective of
finding a good permutation is to enable loop tiling. Since loop
fusion/distribution decisions are also made at this stage, we would want to
model all possible fusion opportunities that enable tiling and pick one of
them. It is a part of our future work, to come up with a cost function that
decides a good fusion / distribution strategy.


\subsection{Definitions}
We first provide some definitions and properties of data structures that we use in
order to model the space of all permutations for fusion and tileability.
The central data structure that we use in modelling all possible fusion
opportunities that enable tiling is the \emph{fusion conflict graph}. 
A fusion conflict graph (FCG), $G=(V,E)$, where
the set of vertices is given by $V =
\{S_1^1,S_1^2,\dots, S_1^{\dim{S_1}}, S_2^1, \dots ,S_n^{\dim{S_n}}\}$, has a
vertex corresponding to each dimension of a statement in the program.
The vertices of the dependence graph are the statements in the
program. Therefore, for every vertex $v$ in the FCG, there exists a statement
(vertex) $S$ in the dependence graph such that $v$ corresponds to a dimension of
$S$. Hence, one can define a function $f:FCG\rightarrow DDG$, from the vertices of the fusion
conflict graph to the vertices of the dependence graph.

The edges in the FCG represents the dimensions of statements that can not be fused together and
permuted to the outermost level. That is, if there exists an edge between
$S^1_i$ and $S^2_j$, then the $i^{th}$ dimension of $S_1$ and $j^{th}$ dimension
of $S_2$ can not be fused together and permuted to the outermost level. Note
that, if the loop nest can be fully permuted, then it can be tiled as well. Hence an
edge in the fusion conflict graph encodes violation of fusion and tileability.
Once the graph is constructed, the objective is to group vertices that are not
connected by edges, without violating any dependences. Independent sets group
vertices in a graph that are not connected by edges. However, in order to
not violate any dependences, these independent sets have to be convex.
Given a fusion conflict graph, we say that an independent set $\mathcal{I}$ of
the fusion conflict
graph is \emph{convex}, if the $\mathcal{I}$ is an independent set and for each $v \in
\mathcal{I}$, the following condition holds:
\begin{equation}
    S=f(v) \wedge \forall S_1 \in \mathit{Pred}(S) \exists v_1 \in \mathcal{I}.
    f(v_1)=S.
\end{equation}
That is, if a vertex $v$ of the FCG corresponding to a statement $S$ is present
in $\mathcal{I}$, then there must be a vertex $v_1$ corresponding to every
predecessor $S_1$ of $S$ in $\mathcal{I}$. This condition is required to encode
transitive dependences across vertices (statements) in the DDG.

We obtain a convex independent set by a \emph{convex} coloring of the FCG. 
Given a fusion conflict graph, we say that the coloring of the fusion conflict
graph is \emph{convex}, if the vertices that have the same color form a convex
independent set.
Note that there can exist many convex colorings for the given FCG. We pick one
of them. Coming up with a cost model that picks a good coloring is a part of
our future work. In the rest of this section, we provide an approach that
performs a constructs an FCG and performs the convex coloring of the FCG. The
number of colors used to color the FCG is bounded by the maximum dimensionality
of the loop nest. This enforces a mapping of colors to dimensions of the loop
nest. The colors are ordered; the ordering of the colors give
the permutation for a statement from the outermost level to the innermost.

%
\subsection{Approach}
\label{ssec:permuteandFuse}
In this section, we provide an algorithm which to find a valid permutation. 
Note that in the space of all valid permutations, one might want to explore different cost
models that capture fusion strategies that maximize performance. For example, a
fusion that does not inhibit parallelism might be desired. In such a case, one
would want a cost model where the loop nests are fused only if the resulting
loop nest is parallel. Coming up with a cost model to enable optimal fusion, is
a focus of our future work. In this paper, we only present an approach that
models the of all possible loop permutations that enable fusion and tiling and
picks one of them.

We propose a two stage approach as shown in
Algorithm~\ref{algo:permuteandFuse}. The first step models the space of all
permutations that enable fusion and permutation by constructing a fusion
conflict graph. The routine \textsc{BuildFCG} in Line~\ref{line:buildFCG} of the
algorithm constructs the fusion conflict graph. The description of this routine
is given in Section~\ref{ssec:fcg-construction}. The number of colors used to
color the FCG is bounded by the maximum dimensionality of the loop nest.

The routine \textsc{ColorFCG} in
Line~\ref{line:convexColour} performs a convex coloring of the FCG. A convex
coloring of the FCG. If the loop nest can be completely fused, then the graph can
be colored with $m_S$ colors, where $m_S$ is the maximum depth of a loop nest
in the program. The vertices that obtain the same color represent the
dimensions that can be fused together and permuted to the outermost level. If
the graph is not colorable with $m_S$ colors, then we find a
subgraph of the FCG which can be colored with $m_S$ colors. If this subgraph is
a maximal subgraph of FCG that can be colored with $m_S$ colors, then
by fusing the dimensions with the same color, we obtain a maximally fused loop
nest. However, finding the maximally colored subgraph is costly. Therefore, as
a trade-off, we do not aim at finding the maximal subgraph which is colorable
with $m_S$ colors. We employ a SCC based coloring algorithm which
colors SCC by SCC for a given color. If coloring fails, we cut the DDG, update
the FCG and then continue coloring. The details of this coloring step is given
in Section~\ref{ssec:coloring}. Assuming there is a mapping from a color
to a dimension of the loop nest, the
permutation for a statement $S$ at a given level can be obtained based on the
colors of the vertices in the FCG that correspond to $S$.

\begin{algorithm}
    \caption{\textsc{PermuteAndFuse(P, G)}}
    \label{algo:permuteandFuse}
    \begin{algorithmic}[1]
    \REQUIRE {Program $P$ and DDG $G$ of $P$}
    \ENSURE {A valid permutation $T$ for each statement in $P$}
        \STATE F $\gets$\textsc{BuildFCG(G)}\label{line:buildFCG}
        \STATE maxColours $\gets$ Maximum dimsionality of a loop nest in $P$
        \STATE \textsc{ColorFCG(F,G,maxColours)}\label{line:convexColour}
    \end{algorithmic}
\end{algorithm}
Algorithm~\ref{algo:permuteandFuse} is sound, that is, the
permutation found for fusion and tileability is valid and does not violate any
dependences. The algorithm is also complete in the sense that all possible
permutations that enable fusion and tiling are present in the space that we
model. Therefore any valid fusion can be chosen by using an appropriate cost
model.

\subsection{Construction of the fusion conflict graph}
\label{ssec:fcg-construction}
In this section, we describe the construction of the fusion conflict graph. 
Recall that, each dimension of a statement in the program has a corresponding
vertex in the
fusion conflict graph. An edge in the fusion conflict graph represents the
dimensions that can not be fused together and permuted to the outermost level.

\begin{algorithm}
    \caption{\textsc{BuildFCG(DDG G)}}
\label{algo:buildFCG}
\begin{algorithmic}[1]
    \REQUIRE {Dependence Graph $G\langle G_v,G_e\rangle$ }
    \ENSURE {Fusion Conflict Graph $F\langle F_v,F_e\rangle$ }
    \FORALL {$S \in G_v$} \label{line:begin-permute-edges}
        \STATE  $\psi \gets$ All intra statement dep constraints for $S$
        \FORALL {$i \in 1 \dots m_S$}
            \IF {$(c^i_S \geq 1) \wedge \psi $ is infeasible}
                \STATE $F_e \gets F_e \cup \{c^i_S \rightarrow c^i_S\}$\label{line:end-permute-edges}
            \ENDIF
        \ENDFOR
    \ENDFOR 

    \FORALL {pair of statements $(S_s, S_t)$ such that $i>j$} \label{line:begin-pairwise-edges}
        \STATE $\psi \gets$ Dep constraints for all deps between $S_s$ and $S_t$
        \FORALL {$i \in 1 \dots  m_{S_s}$}
            \FORALL {$j \in 1 \dots  m_{S_s}$}
                \IF {$(c^i_{S_s} \geq 1 \wedge c^j_{S_t} \geq 1) \wedge \psi$ is infeasible}
                    \STATE $F_e \gets F_e \cup \{c^i_{S_s} \gets c^j_{S_t}\}$
                \ENDIF
            \ENDFOR
        \ENDFOR
    \ENDFOR \label{line:end-pairwise-edges}

    \FORALL {$S \in G_v$}
        \FORALL {$i \in 1 \dots m_S$}
            \STATE $F_e \gets F_e \cup \{c^i_S \rightarrow c^j_S| i\neq j \wedge 1
\leq j \leq m_S\}$\label{line:intraStmtEdge}
        \ENDFOR
    \ENDFOR

    \RETURN $G$
\end{algorithmic}
\end{algorithm}

Algorithm~\ref{algo:buildFCG} builds a fusion conflict graph. It incrementally
adds edges between vertices of the FCG by analyzing dependences between every
pair of statements in the program. The edges are added in two stages: (a) the
first stage adds intra statement permute preventing edges (b) the second
stage that adds inter statement permute and fusion preventing edges.

\paragraph{Adding intra statement edges:}
Given a DDG, the algorithm for each statement $S$, collects all intra statement
dependences as shown in Equation~\ref{eqn:intraStmtDeps}:
\begin{align}
\label{eqn:intraStmtDeps}
D_S \equiv \bigwedge \limits_{e \in DDG} (D_e| \mathit{Src(e) = Dest(e) = S}).
\end{align}
If a dimension $i$ is not permutable to the outermost level, then it
violates at least one of the intra statement dependences when permuted. To find
if $i$ is permutable to the outermost level, we set the lower bound of the
corresponding coefficient, $c^i_S$, to 1. Other coefficients corresponding to
the statement $S$ except $c^0_S$, which corresponds to the shift are set to
zero. Note that all variable corresponding to the transformation coefficients
are not constrained to be integers. 
We then solve pluto-lp for the dependence polyhedron $D_S$. 
If these constraints are unsatisfiable, then the dimension $i$ is not
permutable. Therefore we add a self edge on the vertex $c^i_S$ in the FCG
(Lines~\ref{line:begin-permute-edges}-\ref{line:end-permute-edges}).
Adding a self edge prevents coloring the vertex $i$ indicating that the
dimension can not be permuted. The self edge will removed only when the permute
preventing dependence(s) are satisfied at some outer level.

\paragraph{Adding inter statement edges}
Algorithm ~\ref{algo:buildFCG} adds inter statement permute and fusion
preventing edges in Lines~\ref{line:begin-pairwise-edges}-\ref{line:end-pairwise-edges}. For each pair of
statements $S_s$ and $S_t$ that are connected in the DDG, it collects all
dependence constraints ( both intra and inter statement dependences) between
them as in Equation~\ref{eqn:pairwiseDep}: 
\begin{align}
    \label{eqn:pairwiseDep}
    D \equiv \bigwedge\limits_{e\in DDG} (D_e | \mathit{Src(e),Dest(e)\in
\{S^s,S^t\}}).
\end{align}

The algorithm adds an edge between $S^s_i$ and $S^t_j$ if fusing and
permuting the $i$ and $j$ dimensions of $S^s$ and $S^t$ does not violate any
dependence between $S^s$ and $S^t$. 
Let $c^s_i$ and $c^t_j$ be the coefficients corresponding to $S^s_i$ and $S^t_j$
respectively. After adding the constraints, $c^s_i \geq 1$ and $c^s_j\geq 1$ and
setting all the transformation coefficients corresponding to other dimensions to
zero, if Pluto's LP formulation with $D$ as the dependence polyhedron is
unsatisfiable, then fusing $i^{th}$ dimension of $S^s$ with $j^{th}$ dimension
of $S^t$ will violate a dependence. Hence an edge is added between $S^s_i$ and
$S^t_j$.

Note that during the addition of edges, Algorithm~\ref{algo:buildFCG} uses the
relaxed LP formulation, for a pair of statements. This is a
an application of {\it pluto-lp} on a pair of statements. 
Finally, the algorithm adds edges between the vertices of the FCG that
correspond to the same statement (line~\ref{line:intraStmtEdge}). This ensures that we do
not assign the same color to two dimensions of a statement.

%
%
%

\subsubsection{Conflicting Shifts}
A pairwise loop shifting transformation for a particular dimension $i$ is said
to be \emph{conflicting} for a set of statements $\mathbb{S}$ if there are
valid shifts for the $i^{\mathit{th}}$ dimension for each pair of statements in
$\mathbb{S}$ that allow pairwise fusion and permutability but there does not
exist a valid shift for fusion and permutation for all the statements in
$\mathbb{S}$. In Theorem~\ref{thm:scale-conflicts} we formally prove that a conflicting loop
shifting transformation does not exist.

\begin{theorem}
    \label{thm:scale-conflicts}
    Given a program $P$, the DDG $G$ of $P$, there does not exist a set of
    statements $\mathbb{S}$, such that there is a conflicting loop shifting
    transformation for statements in $\mathbb{S}$. 
\end{theorem}
Before presenting the proof of the Theorem~\ref{thm:scale-conflicts}, we make
the following observation about a loop shifting transformation. Consider a
dependence $S_s \rightarrow S_t$ in $P$.  Let a loop shifting transformation
$[i]\rightarrow[i+1]$ be applied for $S_t$.  This transformation will execute
the $i^\mathit{th}$ iteration in the original space at the $i+1^\mathit{th}$
iteration. That is, the execution is delayed by a factor of $1$ in the
transformed space. This leads to an observation that, any positive loop shifting
transformation at the target of a dependence will not violate the dependence.
    \begin{proof}
        We prove Theorem~\ref{thm:scale-conflicts} by contradiction. We prove split the proof into three cases.
        \begin{enumerate}
            \item \label{proof:diff-scc}
                All the statements considered are from different SCCs. We
                assume that the statements are pairwise fusable and permutable
                for the dimension $i$.  Since the statements are in different
                SCCs, and pairwise fusion of these statements is valid,
                possibly with a loop shift, all the statements can be
                incrementally fused together by increasing the shifting factors
                of the targets of dependences. Therefore, there will exist a
                shifting factor, which will allow us to fuse the
                $i^\mathit{th}$ dimension of all the statements considered.
                This large positive number will be the shift required to fuse
                all the statements and can be obtained by a call to pluto-lp.
                This is guaranteed to have a solution in this case, and the
                resulting  rational solution can be scaled to an integer.
                Thus there can not exist a conflicting which will
                prevent the statements in different SCCs to be fused.

            \item Suppose the statements are in the same SCC, and the
                $i^\mathit{th}$ dimension does not carry the backward
                dependences. In this case, any loop shifting transformation to
                the $i^\mathit{th}$ will not violate a backward dependence.
                Then, for the $i^\mathit{th}$ dimension, all statements can be
                seen as target of the forward dependences. This is similar to
                case~\ref{proof:diff-scc}.

            \item Let the dimension $i$ is the source of some dependences and the
                target of few others. That is, the dimension $i$ carries both
                loop independent and loop carried dependences. Since the statement is in an SCC, there
                exists a dimension which carries the dependence. Without loss
                of generality let us assume that dimension $i$  carries these 
                dependences. These dependences will be satisfied by the dimension $i$
                and any loop shifting transformation will not violate the
                dependences. Therefore, there exists a shifting factors in the original
                schedule which will now violate any of these dependences and
                fuse all the statements within this SCC. Therefore the case of conflicting
                shifts do not arise.
           \end{enumerate}
    \end{proof}

   A similar argument can be provided for the absence of conflicting loop
   scaling factors. 

\subsection{Colouring the FCG}
\label{ssec:coloring}
The routine provided in Algorithm~\ref{algo:buildFCG} builds the fusion
conflict graph. In this section, we provide a routine that implements a convex
coloring of the FCG. Coloring is driven by the topological ordering on the
SCCs of the DDG.

Algorithm~\ref{algo:colorSccBased} colors the vertices of the FCG, one color
at a time. Coloring starts from the first SCC in the topological order of the
SCCs. Given an SCC, there exists at least one dimension which fuses all the
vertices of the SCC without violating any dependences. This observation has
formally been stated in Theorem~\ref{thm:scc_coloring}.
\begin{theorem}
    \label{thm:scc_coloring}
    Given a DDG $G$ of a program P, and a FCG $F$ of $P$, then for each SCC in
    G, there exists at least one dimension corresponding to vertices of the SCC,
    such that the vertices corresponding to these dimensions in $F$ can be given the
    same color.
\end{theorem} 
\begin{proof}
    We prove Theorem~\ref{thm:scc_coloring} by contradiction. Let us assume
    that there existed no dimensions for every vertex in the SCC which can be
    given the same color. This means to say that there exists no
    dimensions that can be fused together and permuted to the outermost level.
    However, in the input program, since the vertices are part of an SCC, there
    must exist a loop, which carries a dependence along the back edge
    (Otherwise, these statements would not have been in an SCC).
    Therefore, when the pair of statements are analyzed using
    pluto-lp, at least one dimension corresponding to the loop that fuses all
    the statements will have a solution to pluto-lp. Therefore there will not be
    any edges added corresponding to the dimensions that carry the dependences for all the
    statements in the FCG. Therefore, all the vertices corresponding to the
    at least one dimension can be given the same color.
\end{proof}

In the outermost level, the coloring of
the first SCC will succeed. The coloring of a subsequent SCCs might fail, if
fusing it with one of the SCCs that has already been colored violates a
dependence. Note that in such a case there will be an edge between the
vertices of the FCG corresponding to the statements of two SCCs that can not be
fused. As soon as coloring fails, the algorithm cuts between the two SCCs and
updates the DDG. FCG is also updated to remove edges corresponding to the
dependences that have already been satisfied. Note that at the inner levels,
the coloring of the first SCC might fail because of a permute preventing
dependence. This dependence must be satisfied at some outer level. In such
cases, we update the DDG by removing the edges corresponding to the dependences
that have been satisfied at the outer levels. We also rebuild the FCG in order
to remove edges that correspond to the dependences that have been satisfied at
outer levels.

%
\begin{algorithm}[ht]
    \caption{\textsc{colorFCG(FCG, DDG, maxColors)}}
    \label{algo:colorSccBased}
    \begin{algorithmic}[1]
        \FORALL {$c \in 1 \dots \mathit{maxColors}$} \label{line:loop-color}
        \FORALL{$i = 1$ to $\card{SCCs(DDG)}$ } \label{line:traverse-scc}
    \IF{$\neg$\textsc{colorScc}(i,c,FCG)}\label{line:colorScc}
        \IF{$i==1$}\label{line:permuteDep}
                        \STATE Update DDG by removed deps satisfied at outer levels
                \ELSE
                        \STATE \textsc{cutBetweenSCCs}(i, i-1, DDG)\label{line:fusionDep}
                        \STATE Update DDG by removing dependences satisfied by the cut and
                               other outer levels.
                \ENDIF
                \STATE \textsc{BuildFCG(DDG)}\label{line:rebuildFCG}
            \STATE \textsc{color}(i,c,DDG)\label{line:re-colorSCC}
        \ENDIF
        \ENDFOR
        \FORALL {$j=1 \dots \card{S}$}
        \STATE Update the $T_{S_j}$ at level $c$ for based on the vertices of the FCG that have been colored
             with $c$ \label{line:update-permutation}
        \ENDFOR
        \ENDFOR
    \end{algorithmic}
\end{algorithm}

Algorithm~\ref{algo:colorSccBased} assumes that all there is an ordering on the
colors and the SCC of the program. For a given color $c$, algorithm picks SCCs in
the topological order (line~\ref{line:traverse-scc}). It tries to color the
vertices of the FCG corresponding to the statements in the current SCC $i$.
This is accomplished by the routine \textsc{colorSCC}
(line~\ref{line:colorScc}). This routine returns true if the coloring succeeds
for the SCC $i$. Else it returns false. When the coloring fails, the cases are
handled in Lines~\ref{line:permuteDep} and~\ref{line:fusionDep}. The DDG is
either directly updated or updated after the cut, by removing dependences that
are satisfied by the outer levels and the cut, based on the level at which
coloring fails. FCG is rebuilt once the
dependencies are removed (line~\ref{line:rebuildFCG}) and the current SCC $i$ is
colored again with color $c$ (line~\ref{line:re-colorSCC}). Finally the
permutation at the level $c$ is updated for all the statements in the program
(line~\ref{line:update-permutation}).

\subsection{Illustration:}
Consider the example shown in Figure~\ref{fig:example-interchange}.
Note that the statements can not be fused directly because of the RAW
dependence
from statement S1 to S2. But if the dimensions of $S_1$ are
interchanged, which corresponds to loop interchange, then $S_1, S_2$ and $S_3$ can be
fused completely. 
The FCG corresponding to the program is shown in the right. There are six vertices
each corresponding to a dimension of a statement of the loop nest. There are no
intra statement dependences and hence no permute preventing edges are added to
the FCG in Lines~\ref{line:begin-permute-edges}-~\ref{line:end-permute-edges}.
Thick lines in the figure represent the edges added by the intra statement
dependences. Edges added between the dimensions of a
statement are shown as dashed lines. 

\begin{figure}[ht]
    \centering{
        \begin{minipage}{0.55\linewidth}
        \begin{lstlisting}[
            basicstyle =\footnotesize
        ]
for(i=0; i<N; i++)
  for(j=0;j<N;j++)
    A[i][j] = i+j; //S1
for(i=0; i<N; i++)
  for(j=0;j<N;j++)
    B[i][j] = A[j][i]; //S2
for(i=0; i<N; i++)
  for(j=0;j<N;j++)
    C[i][j] = A[i][j]; //S3
    \end{lstlisting}
    \end{minipage}\hfill
    \begin{minipage}{0.4\linewidth}
        \input{tikz/example-interchange.tex}
    \end{minipage}
}
\caption{Illustration of loop interchange to enable fusion}
\label{fig:example-interchange}
\end{figure}
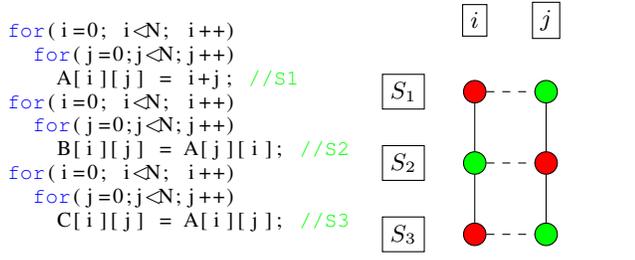
The coloring done by the Algorithm~\ref{algo:colorSccBased} is shown in the
Figure~\ref{fig:example-interchange}. 
The vertices in red form the outer most level and the ones in
green form the inner level. This corresponds to identity transformation for
statements $S_1$ and $S_3$ and a loop interchange for the statement $S_2$ thus
resulting in a fully fused loop nest.

\subsection{Correctness}
The correctness of the approach of finding a valid transformation depends on
the transitivity of fusion of dimensions while satisfying the tileability
criterion. The fusion conflict graph is constructed by analyzing pairwise and
hence in the following theorem, we formally prove that dimension wise fusion and
permutability is transitive.
\begin{theorem}
    \label{thm:transitivity}
    If a dimension $i$ of statement $S_1$ can be fused together with a
    dimension $j$ of a statement $S_2$ and dimension $j$ of $S_2$ can be fused
    and permuted with a dimension $k$ of statement $S_3$ then dimensions $i,j$
    and $k$ of statements $S_1,S_2$ and $S_3$ can be fused together and
    permuted provided there are no loop skewing transformations.
\end{theorem}

\begin{proof}
We will prove Theorem~\ref{thm:transitivity} by contradiction. Let us assume
that the dimensions $i$ and $j$ of $S_1$ and $S_2$ can be fused together and
dimensions $j$ and $k$ of $S_2$ and $S_3$ can be fused together, but all the
three statements can not be fused. Then there exists at least one dependence
which is violated. Let this dependence be $d$. The dependence $d$ can not be a
direct dependence between $S_1$ and $S_3$ because, if this was a direct
dependence, there would have been a fusion conflict edge between dimensions $i$
and $k$ of statements $S_1$ and $S_3$ in the fusion conflict graph. This fusion
does not violate the transitive dependence between $S_1$, $S_2$ and $S_3$
because the satisfaction of relaxed-LP formulation of Pluto implies that each
of the dimensions can be fused and permuted to the outer most level. This
implies that once the loops i and j form permutable band and the loops j and k
form a permutable band. Therefore, the loops i,j and k of statements $S_1$,
$S_2$ and $S_3$ can be fused together to form a permutable band. This will not
violate any dependences between statements $S_1$, $S_2$ and $S_3$. This violates
our assumption that there is a dependence that is violated. Hence $S_1$, $S_2$
and $S_3$ can be fused together and permuted to the outermost level.
\end{proof}

%% file: tikz/example-interchange.tex
 \begin{tikzpicture}[align=center, node distance=3em]
\node[draw, circle, fill=red] (A) at (0,0) {};
\node[draw, circle, fill=green] (B) [right of = A] {};
\node[draw, circle, fill=green] (C) [below of = A] {};
\node[draw, circle, fill=red] (D) [below of = B] {};
\node[draw, circle, fill=red] (E) [below of = C] {};
\node[draw, circle, fill=green] (F) [below of = D] {};

\node[draw, align=left] (T1) [left of=A] {$S_1$};
\node[draw, align=left] (T2) [left of=C] {$S_2$};
\node[draw, align=left] (T3) [left of=E] {$S_3$};

\node[draw, align=left] (T4) [above of=A] {$i$};
\node[draw, align=left] (T5) [above of=B] {$j$};

\draw  (A)--(C);
\draw[dashed]  (A)--(B);
\draw  (B)--(D);
\draw[dashed]  (C)--(D);
\draw  (C)--(E);
\draw  (D)--(F);
\draw [dashed] (E)--(F);
\end{tikzpicture}